%% file: RR-BGP.tex
\documentclass{elsarticle}
\usepackage{hyperref}

\usepackage{mathrsfs,color,amsmath,amssymb,amsthm}

\DeclareMathAlphabet{\mathbbold}{U}{bbold}{m}{n}                    
\newcommand{\zero}{\mathbbold{0}}
\newcommand{\unit}{\mathbbold{1}}
\usepackage{epsfig}
\begin{document}
\newcommand{\scap}{\operatornamewithlimits{\sqcap}}
\newcommand{\scup}{\operatornamewithlimits{\sqcup}}
\newcommand{\floor}[1]{\left\lfloor#1\right\rfloor}
\newcommand{\ceil}[1]{\left\lceil#1\right\rceil}
\newcommand{\nin}{\mrm{\it in}}
\newcommand{\nout}{\mrm{\it out}}
\newcommand{\comment}[1]{{\footnotesize\tt #1}}
\newtheorem{theorem}{Theorem}
\newtheorem{proposition}{Proposition}
\newtheorem{lemma}{Lemma}
\newtheorem{corollary}{Corollary}
\newtheorem{definition}{Definition}
\def\<#1,#2>{\langle #1,#2\rangle}
\newcommand{\cala}{\mathscr{A}}
\newcommand{\calc}{\mathscr{C}}
\newcommand{\calp}{\mathscr{P}}
\newcommand{\cals}{\mathscr{S}}
\newcommand{\ev}[2]{[#2]_{#1}}
\newcommand {\cqfd}     {\hfill $\Box$}
\newcommand{\set}[2]{\{#1\mid\,#2\}}
\newcommand{\mrm}[1]{\text{\rm #1}}
\newcommand{\lcm}{\operatorname{lcm}}
\newcommand{\ov}[1]{\overline{#1}}
\newcommand{\Q}{\mathbb{Q}}
\newcommand{\qbar}{\Q\cup\{-\infty\}}
\newcommand{\R}{\mathbb{R}}
\newcommand{\Z}{\mathbb{Z}}
\newcommand{\N}{\mathbb{N}}
\newcommand{\B}{\mathbb{B}}
\newcommand{\qmax}{\Q_{\max}}
\def\rbar{$\rm \bf \bar{R}$ }
\newcommand{\sfA}{\mathsf{A}}
\newcommand{\sfK}{\mathsf{K}}
\newcommand{\sfa}{\mathsf{a}}
\newcommand{\sfb}{\mathsf{b}}
\newcommand{\sfc}{\mathsf{c}}
\newcommand{\sfd}{\mathsf{d}}
\newcommand{\sfp}{\mathsf{p}}
\newcommand{\sfP}{\mathsf{P}}
\newcommand{\sfr}{\mathsf{r}}
\newcommand{\sfu}{\mathsf{u}}
\newcommand{\sfv}{\mathsf{v}}
\newcommand{\sfq}{\mathsf{q}}
\newcommand{\sfs}{\mathsf{s}}
\newcommand{\sfS}{\mathsf{S}}
\newcommand{\sfT}{\mathsf{T}}
\newcommand{\sfU}{\mathsf{U}}
\newcommand{\sfV}{\mathsf{V}}
\newcommand{\sft}{\mathsf{t}}
\newcommand{\sfy}{\mathsf{y}}
\newcommand{\calo}{\mathcal{O}}

\author[cesame]{Vincent Blondel \fnref{fn1}}
\ead{blondel@inma.ucl.ac.be}
\author[ensta]{St\'ephane Gaubert \fnref{fn1}}
\ead{Stephane.Gaubert@inria.fr}
\author[ensl]{Natacha Portier \fnref{fn1,fn2}}
\ead{Natacha.Portier@ens-lyon.fr}

\fntext[fn1]{This work was partly supported by a grant
{\sc Tournesol} (Programme de coop\'eration scientifique entre la France
et la communaut\'e Fran\c{c}aise de Belgique), 
and by the European Community Framework IV program through the
research network ALAPEDES (``The Algebraic Approach to Performance
Evaluation of Discrete Event Systems'').} 

\fntext[fn2]{This work was partially funded by European Community under contract PIOF-GA-2009-236197 of the 7th PCRD.}

\address[cesame]{Large Graphs and Networks, D\'epartement d'ing\'enierie math\'ematique, Universit\'e
catholique de Louvain, 4 Avenue Georges Lema\^\i tre, 
B-1348 Louvain-la-Neuve,
Belgique}
\address[ensta]{INRIA and CMAP, Ecole Polytechnique,
91128 Palaiseau Cedex, France.}
\address[ensl]{LIP, UMR 5668, ENS de Lyon -- cnrs -- UCBL -- INRIA,
\'Ecole Normale Sup\'erieure de Lyon, Universit\'e de Lyon, 46, all\'ee d'Italie,
69364 Lyon cedex 07, France
and Department of Computer Science, University of Toronto, Canada.}

\title{The set of realizations of a max-plus linear sequence
is semi-polyhedral}

\begin{abstract}
We show that the set of realizations of a given dimension
of a max-plus linear sequence is a finite union
of polyhedral sets, which can be computed
from any realization of the sequence. 
This yields an (expensive) algorithm
to solve the max-plus minimal realization
problem. These results are derived
from general facts on rational expressions
over idempotent commutative semirings:
we show more generally that the set of
values of the coefficients of a commutative rational expression
in one letter that yield a given max-plus linear sequence
is a finite union of polyhedral sets.
\end{abstract}

\maketitle

\vspace{2em}

\begin{center}
\framebox[1.1\width]{Research Report RRLIP 2010-33}
\end{center}

\vspace{2em}

\section{Introduction and Statement of Results}
A {\em realization} of a sequence $S_0,S_1,\ldots$ of elements
of a semiring $K$ is a triple $(c,A,b)$,
where $c\in K^{1\times N}$, $A\in K^{N\times N},
b\in K^{N\times 1}$, and $S_0=cb$, $S_1=cAb$, 
$S_2=cA^2b$, \ldots
The integer $N$ is the {\em dimension}
of the realization. A sequence $S$ is $K$-{\em recognizable}
(or $K$-linear) if it has a realization $(c,A,b)$, and
then, we say that $S$ is {\em recognized} by $(c,A,b)$.

In this paper, we consider the max-plus semiring $K=\qmax$,
which is the set $\qbar$, equipped
with the addition $(a,b)\mapsto a\oplus b=\max(a,b)$
and the multiplication
$(a,b)\mapsto a\otimes b=a+b$. 
We address the following
realization problem, 
which was raised as an open problem in several works~\cite{cohen85b-inria,olsder86b,BICOQ,olsder}:
{\em does a $\qmax$-recognizable 
sequence have a realization of a given dimension?}

As observed by the first and third authors~\cite{Blo99},
it follows from an old result of Stockmeyer and
Meyer~\cite{meyer} that this problem is co-NP-hard.
In this paper, we show that it is decidable and we show how one can
effectively construct the set of realizations.
Our results are also valid for other tropical semirings~\cite{pin95},
like the semiring of max-plus integers $\Z_{\max}=
(\Z\cup\{-\infty\},\max,+)$, or the 
semiring $\N_{\min}=(\N\cup\{+\infty\},\min,+)$,
hence, it is convenient to
consider more generally a semiring $K$,
whose addition, multiplication, zero element, and unit
elements will be denoted by $\oplus,\otimes,\zero,\unit$, respectively.
We shall assume that $K$ is commutative, i.e., that $u\otimes v=v\otimes u$. We shall use the familiar
algebraic notation, with the obvious
changes (e.g., $a^2b=a\otimes a\otimes b$).

We say that a semiring
is {\em idempotent} when $u\oplus u=u$,
we say that an idempotent semiring is {\em linearly
ordered} when the relation $u\leq v\iff u\oplus v=v$
is a linear order, and that it is {\em archimedian}
if $u\lambda^k\geq v\mu^k$
for all $k\geq 0$ implies $v=\zero$ or $\lambda \geq \mu$.
Finally,
we say that $K$ is {\em cancellative}
if $uv=u'v\implies v=\zero$ or $u=u'$.

A {\em monomial} in the $n$ variables
$x_1,\ldots,x_n$ is of the form
$m(x)= u x_1^{\alpha_1}\cdots x_n^{\alpha_n}$,
for some $u\in K$ and $\alpha_1,\ldots,\alpha_n\in \N$.
We call {\em half-space}
of $K^n$ a set of the form $\set{x\in K^n}{m(x)\geq m'(x)}$,
where $m$ and $m'$ are monomials.
(In $\qmax$, a monomial can be rewritten with the
conventional notation as $m(x)= u+ \sum_{i=1}^n \alpha_i x_i $,
which accounts for the terminology ``half-space''). A {\em polyhedron}
is a finite intersection of half-spaces.
A set is {\em semi-polyhedral} if it is 
a finite union of polyhedra.

A realization of dimension
$N$, $(c,A,b)$, can be seen as an element of the
set $K^{2N+N^2}$. We will prove:

\begin{theorem}\label{th-alg}
Let $K$ denote an idempotent linearly ordered archimedian
cancellative commutative semiring. Then, the set
of realizations of dimension $n$ of a $K$-recognizable
series is a semi-polyhedral subset of $K^{2N+N^2}$,
which can be effectively constructed.
\end{theorem}

We get as a consequence
of Theorem~\ref{th-alg}:

\begin{corollary}
\label{cor1}
\sloppy
When $K=\qmax$, $\Z_{\max}$, or $\N_{\min}$,
the existence of a realization of dimension
$n$ of a $K$-recognizable sequence is decidable.
\end{corollary}

Indeed, when $K=\qmax$, the non-emptiness of a 
semi-polyhedral set is decidable,
because the first order theory of $(\Q,+,\leq)$
is decidable, or, to use a perhaps more elementary
argument, because the non-emptiness of an ordinary polyhedron
can be checked by linear programming (see e.g.~\cite{schrijver}).
When $K=\Z_{\max}$
or $\N_{\min}$, the corollary follows from
the decidability of Presburger's arithmetics (see e.g.~\cite{enderton}).

It follows from Corollary~\ref{cor1} that
there is an algorithm to compute max-plus minimal
realizations, a problem which arose from the beginning
of the development of the max-plus modelling of discrete
event systems~\cite{cohen89a}, which was mentioned in the book~\cite{BICOQ}
and was stated by Olsder and De Schutter~\cite{olsder} as one
of the open problems of~\cite{bv}.
In fact, the algorithm is very expensive
(see the discussion in~\S\ref{sec-universal}),
so our result only implies that 
we can solve the realization problem in ``small''
dimension. A Caml implementation by G. Melquiond and P. Philipps
is available~\footnote{http://perso.ens-lyon.fr/natacha.portier/realisations-max-plus.tar.gz}. 
It would be interesting to find
a less expensive algorithm.

Before proving Theorem~\ref{th-alg},
it is instructive to show why classical arguments fail
to prove these result. A natural idea,
would be to show that if two sequences $S$ and $T$
have realizations of respective sizes $N$ and $M$, 
there is an integer $\nu(N,M)$ such that:

\begin{equation}
(S_k=T_k ,\forall k\leq \nu(N,M)) \implies (S_k=T_k,\forall k\in \N) \enspace.
\label{eq-class}
\end{equation}

(Results of this kind are called ``equality theorems''
by Eilenberg, see~\cite[Chap.~6, \S~8]{eilenberg74}.)
Indeed, if the semiring $K$ satisfies property~\eqref{eq-class},
then, the set of realizations of dimension $N$
of a sequence $T$ given by a realization of dimension
$M$ is the set defined by the finite system of equations $cA^kb=T_k$,
for $k=0,\ldots, \nu(N,M)$. There are two classical cases
where property~\eqref{eq-class} is true.
First, if $K$ is a finite semiring (like
the Boolean semiring),~\eqref{eq-class}
is trivially true since the set of realizations
of a given dimension is finite (and,
of course, the minimal realization problem is decidable).
A second, more interesting case, is when $K$
is a subsemiring of a commutative ring. Then, the Cayley-Hamilton
theorem implies
that~\eqref{eq-class} holds with $\nu(N,M)=N+M-1$,
by a standard argument (see~\cite[Chap.~6, proof of Th.~8.1]{eilenberg74}).
An interesting feature 
of the max-plus semiring is
that~\eqref{eq-class} does not hold. For instance, the realization 
of dimension $2$ over $\qmax$,
\[
c=\begin{pmatrix} 0 & 0 \end{pmatrix},\quad 
A=\begin{pmatrix} 0 & -\infty \\ 
-\infty & -1  \end{pmatrix},\quad 
b= \begin{pmatrix} \alpha \\  0 \end{pmatrix},\quad 
\]
where $\alpha$ is an element of $\qmax$, 
recognizes the sequence $S^\alpha: S^\alpha_k = \max(\alpha,-k)$. To
distinguish between $S^\alpha$ and $S^\beta$, we need to consider
values of $k\geq \min(-\alpha,-\beta)$, 
and this contradicts~\eqref{eq-class}.


Our proof of Theorem~\ref{th-alg} relies on
a more general result, of independent interest.
Let us first briefly recall some basic
facts about rational series
in one letter (see~\cite{berstelreut} for a detailed
presentation).
Let $X$ denote an indeterminate.
A sequence $S_0,S_1,\ldots\in K$ can be identified
to the formal series $S=S_0\oplus S_1 X\oplus S_2X^2 \oplus\cdots
\in K[[X]]$ (in particular, the indeterminate
$X$ corresponds to the sequence $\zero,\unit,\zero,\zero,\ldots$).
The set of formal series $K[[X]]$, equipped with
entrywise sum and Cauchy product, is a semiring.
The Kleene's star of a series $S$,
defined when $S$ has a zero
constant coefficient, is $S^*=S^0\oplus S\oplus S^2 \oplus
\cdots$ The $k$-th coefficient of $S$
will sometimes be denoted by $\<S,X^k>$ instead
of $S_k$. The Kleene-Sch\"utzenberger theorem states
that $S$ is recognizable if, and only if, it is
rational, i.e., if it can be represented
by a well formed expression involving sums, products, stars,
and monomials. 

Consider now a finite
set of commuting indeterminates,
$\Sigma=\{\sfd_1,\ldots,\sfd_n\}$,
and let $K[\Sigma]$ denote the semiring
of polynomials in $\sfd_1,\ldots,\sfd_n$.
To a vector $d=(d_1,\ldots,d_n)\in K^n$,
we associate the evaluation morphism 
$K[\Sigma][[X]]\to K[[X]]$,
which sends the series $\sfS\in 
K[\Sigma][[X]]$ to the series
$\ev d \sfS$ obtained by
replacing each indeterminate $\sfd_i$ by the value $d_i$.
Borrowing the probabilist notation,
we denote by $\{\sfS=S\}$ the set
$\set{d\in K^n}{\ev{d}{\sfS}=S}$.
More generally, for $\sfS,\sfT\in K[\Sigma][[X]]$,
we shall write for instance $\{\sfS\geq \sfT\}$
as an abbreviation of $\set{d\in K^n}{\ev d{\sfS}\geq \ev d{\sfT}}$.

\begin{theorem}\label{th-2}(Rational series synthesis)
Let $K$ denote an idempotent linearly ordered archimedian
cancellative 
commutative semiring. For all rational series $\sfS\in K[\Sigma][[X]]$
and for all rational series $S\in K[[X]]$, the set 
$\{\sfS=S\}$ is semi-polyhedral.
\end{theorem}

This theorem will be proved in Section~\ref{proof-section}. 

An intuitive way to state this result
is to say that ``the set of values of the coefficients
of a rational expression which yield a given rational
series is semi-polyhedral''. 

Theorem~\ref{th-alg} is an immediate corollary
of Theorem~\ref{th-2}.
Indeed, consider the set $\Sigma=\Sigma_N$ 
whose elements are the $2N+N^2$ indeterminates
$\sfc_{i},\sfA_{ij},\sfb_{j}$, where $1\leq i,j\leq N$.
Let $\sfc=(\sfc_i)\in (K[\Sigma_N])^{1\times N}$, 
$\sfA=(\sfA_{ij})\in (K[\Sigma_N])^{N\times N}$, 
$\sfb=(\sfb_j)\in (K[\Sigma_N])^{N\times 1}$, 
and consider the {\em universal series}
$\sfS_N = \sfc (\sfA X)^*\sfb = \sfc \sfb \oplus \sfc\sfA\sfb X
\oplus\cdots \in K[\Sigma_N][[X]]$, which, by construction, is recognizable
(or equivalently, rational).
Since the set of realizations
of dimension $n$ of a rational series $S\in K[[X]]$ 
is exactly $\{\sfS_N=S\}$,
Theorem~\ref{th-2} implies Theorem~\ref{th-alg}.

We warn the reader that some apparently minor variants
of $\{\sfS =S\}$ need not be semi-polyhedral. For
instance, since $\{\sfS\leq S\}=
\{\sfS\oplus S=S\}$, by Theorem~\ref{th-2},
$\{\sfS\leq S\}$ is semi-polyhedral,
but we shall see in~\S\ref{sec-ex}
that $\{\sfS\geq S\}$ need not be semi-polyhedral.

In Section~\ref{sec-universal}, we bound the complexity
of the algorithm which is contained in the proof of Theorems~\ref{th-alg}
and~\ref{th-2}.
The details of this complexity analysis are lengthy,
but its principle is simple: we need first to compute a 
star height one representation of the universal series
$\sfS_N = \sfc (\sfA X)^*\sfb$. We give an explicit
representation, which turns out to be of double exponential
size. Then, we compute the semi-polyhedral set arising from
this expression, which yields a simply exponential blow up,
leading to a final triple exponential bound. 

This high complexity implies that Theorem~\ref{th-alg}
is only of theoretical interest. However, it should be noted
that Theorem~\ref{th-2} allows us to solve
more generally the ``structured
realization problem'', in which some
coefficients of the realizations are constrained
to be zero. Consider
for instance the problem of computing
all $N$ dimensional realizations $(\sfc, \sfA,\sfb)$
of a linear sequence $S$, subject to the constraint
that $\sfA$ is diagonal. The set
of realizations becomes $\{\bigoplus_{1\leq i\leq N} \sfc_i (\sfA_{ii} X)^*\sfb_i= S\}$, and Theorem~\ref{th-2} shows that this
set is semi-polyhedral. For such structured problems in which the
universal series $\sfS_N$ is replaced by a polynomial size rational
expression, the present approach leads only to a simply exponential complexity.

The algorithmic difficulties encountered here are consistent
with the observation that algorithmic issues concerning linear systems
over rings (and a fortiori over semirings)
are generally harder than in the case over fields.
In particular, the powerful ``geometric approach''
based on the computations of invariant spaces
does carry over to the ring case~\cite{BasMar91},
and even to the max-plus case~\cite{ccggq99,katz07,invariant09},
but then, the analogues of the classical fixed point
algorithms do not always terminate (due to the lack of Artinian
or Noetherian properties). The present algebraic approach, via
rational series, yields alternative tools to the geometric
approach: no termination issue arises, but the algorithms
are subject to a curse of complexity.

It is also instructive to look at Theorem~\ref{th-2} in the light
of the recent developments of tropical geometry~\cite{itenberg,RGST}. The latter
studies in particular the tropical analogues of algebraic sets. 
The tropical analogues of {\em semi}-algebraic sets could
be considered as well: it seems reasonable to define
them precisely as the special semi-polyhedral sets introduced here 
(recall that the exponents appearing in the monomials are
required to be nonnegative integers). 
Then, Theorem~\ref{th-alg} may be thought of as the max-plus
analogue of a known result, that the set of {\em nonnegative} realizations of
a given dimension of a linear sequence over the real numbers (equipped with the usual addition and multiplication) is semi-algebraic (this follows readily
from the ``equality theorem'' mentioned above). 
Then, a comparison with the complexity of existing semi-algebraic
algorithms~\cite{roy} suggests that the present triple exponential
bound is probably suboptimal. To improve it, we would need
to further exploit the tropical semi-algebraic structure.
This raises further issues which are beyond the scope of this paper.

Let us complete this long introduction by pointing
out a few relevant references about the minimal realization
problem. 

First, there are two not so well known
theorems, which hold in arbitrary semirings. 
A result of Fliess~\cite{fliess75b}
characterizes the minimal
dimension of realization as the minimal dimension
of a semimodule stable by shift and containing the
semimodule of rows of the Hankel matrix. 
(The result is stated there for the semiring $(\R^+,+,\times)$,
but, as observed by Jacob~\cite{jacob}, the proof
is valid in an arbitrary semiring.)
Maeda and Kodama found independently closely
related results~\cite{maeda}.
As observed by Duchamp and Reutenauer (see Theorem~2 in~\cite{duchamp}),
Fliess's characterization is a third fundamental statement to add
to the Kleene-Sch\"utzenberger theorem. 
The classical realization theorems over fields
are immediate corollaries of this result.
The results of Anderson, Deistler, Farina
and Bevenuti~\cite{farina96} and Benvenuti and Farina~\cite{farina99}
for nice applications of these ideas.
We also refer the reader to the book~\cite{berstreut} for a general
discussion of minimization issues concerning noncommutative rational
series. 
A second fundamental result, due to Eilenberg~\cite[Ch.~16] {eilenberg74}
(inspired by Kalman), extends the notion of recognizability
and shows the existence of a minimal module which recognizes
a sequence. The difficulty is that this module
need not be free. (Eilenberg's theorem
is stated for modules over rings, but, as noted in~\cite{jpq2},
it can be extended to semimodules over semirings).
The max-plus minimal realization problem was raised by Cohen, 
Moller, Quadrat and Viot~\cite{cohen85b-inria},
and by Olsder~\cite{olsder86b}
(see also~\cite{BICOQ}). There are relatively few
general results about this (hard) problem. 
Olsder~\cite{olsder86b} showed some connections
between max-plus realizations, and conventional
realizations, via exponential asymptotics.
Cuninghame-Green~\cite{cuning91} gave a realization
procedure, which yields, when it can be applied,
an upper bound for the minimal dimension of realization.
Some lower and upper bounds 
involving various notions of rank over the max-plus semiring
were given in~\cite[Chap.~6]{gaubert92a}. In particular, the
cardinality of a minimal generating family of the 
row or column space of the Hankel matrix, which characterizes
the minimal dimension of realization in the case of fields,
is only a (possibly coarse) upper bound in the max-plus semiring.
The lower bound of~\cite[Chap.~6]{gaubert92a} (which
involves max-plus determinants) also appears in~\cite{gbc}, where
it is used to extend to the convex case a theorem proved by 
Butkovi\v{c} and Cuninghame-Green~\cite{cgpb} in the strictly convex case.
De Schutter and De Moor~\cite{deschut} 
observed that the (much simpler) {\em partial}
max-plus realization problem can be interpreted as an extended linear
complementarity problem. This work was pursued by
De Schutter in~\cite{deschutter96}.

\section{Max-plus Rational Expressions}
In this section, we recall some basic results
about max-plus rational expressions, which
will be needed in the proof of Theorem~\ref{th-2}.

The first step of the proof of Theorem~\ref{th-2} is the following
well known star height one representation
(some variants of which
already appeared in particular in works of Moller~\cite{moller88},
of Bonnier-Rigny and Krob~\cite{krob93a},
and in~\cite{gaubert92a,gaubert94a}).
All these results can be thought of as specializations, or refinements,
of general results on commutative rational 
expressions~\cite{eilenberg69,conway71}.

In the sequel, $\sfK$ denotes a generic
semiring (which may or may not coincide with the
semiring $K$ of Theorem~\ref{th-2}).
\begin{lemma}\label{lemma-2}
Let $\sfK$ be an idempotent
commutative semiring. A rational
series $\sfS\in \sfK[[X]]$ can be written
as 
\begin{equation}
\sfS = \bigoplus_{1\leq i\leq r} \sfP_i (\sfq_i X^c)^* 
\label{ratu} \enspace,
\end{equation}
where $\sfP_1,\ldots, \sfP_r\in \sfK[X]$,
$\sfq_1,\ldots,\sfq_r\in \sfK$,
and $c$ is a positive integer.
\end{lemma}
\begin{proof}
It suffices to check that the set of series
of the form~\eqref{ratu} is closed by sum,
Cauchy product, and Kleene's star.
This follows easily from
the following classical
commutative rational identities (see e.g.~\cite{conway71}),
which are valid for all $\sfU,\sfV\in \sfK[[X]]$
(with zero constant coefficient) and $k\geq 1$,
\begin{eqnarray}
(\sfU \oplus \sfV)^* &=& \sfU^*\sfV^* \enspace ,\label{eq-1} \\
(\sfV\sfU^*)^* & =&\unit \oplus \sfV(\sfU \oplus \sfV)^*\enspace,\label{eq-111}\\
\sfU^*&=&(\unit  \oplus \sfU \oplus \cdots \oplus \sfU^{k-1}) (\sfU^k)^*
\label{eq-cyclic}
\end{eqnarray}
(only in~\eqref{eq-1} we used the idempotency and
commutativity of the semiring).
\end{proof}
The representation~\eqref{ratu} of $\sfs$ is far from being unique.
In particular, thanks to the rational identity
\begin{equation}
\sfU^*=\unit \oplus \sfU \oplus \cdots \oplus \sfU^{k-1}\oplus \sfU^k \sfU^* \enspace,
\label{ur}
\end{equation}
which holds for all $k\geq 1$, we can always rewrite
the series~\eqref{ratu} as 
\begin{equation}
\sfS = \sfP\oplus X^{\kappa c}
\Big(\bigoplus_{1\leq i\leq \rho} \sfu_i X^{\mu_i} (\sfq_i X^c)^* 
\Big)
\label{ratu2}
\end{equation}
where $0\leq \mu_i\leq c-1$,
$\sfu_i\in \sfK$, and $\sfP\in \sfK[X]$ has degree
less than $\kappa c$. The interest of~\eqref{ratu2},
by comparison with~\eqref{ratu},
is that the asymptotics of $\<\sfS,X^k>$ can be read directly
from the rational expression. Indeed, for all $0\leq j
\leq  c-1$ and $k\geq 0$,
\begin{equation}
\<\sfS,X^{(k+\kappa)c+j}> = \bigoplus_{\mu_i=j}
\sfu_i\sfq_i^k \enspace. \label{nice}
\end{equation}
When $\sfK$ is the max-plus semiring,
the representations~\eqref{ratu2}
and~\eqref{nice} can be simplified
thanks to the archimedian property.
We say that a series $S\in K[[X]]$
is {\em ultimately geometric} 
if there is an integer $\kappa$ and a scalar $\gamma\in K$
such that $\<S,X^{k+1}>=\gamma \<S,X^k>$,
for all $k\geq \kappa$. The {\em merge}
of $c$ series $S^{(0)},\ldots,S^{(c-1)}$
is the series $S^{(0)}(X^c)
\oplus XS^{(1)}(X^c)
\oplus \cdots X^{c-1}S^{(c-1)}(X^c)$,
whose coefficient sequence is obtained by ``merging''
the coefficient sequences of $S^{(0)}$,
\ldots, $S^{(c-1)}$.
E.g., the merge of 
\begin{equation}
S^{(0)}=X^*=0\oplus 0X\oplus 0X^2
\oplus \cdots
\,\mrm{and}\; S^{(1)}=1(1X)^*=1\oplus 2X\oplus 3X^2
\oplus\cdots
\label{eq-22}
\end{equation} is
\begin{equation}
T=(X^2)^*\oplus 1X(1X^2)^*=
0\oplus 1X\oplus 0X^2\oplus 2X^3 \oplus 0X^4
\oplus 3X^5\oplus\cdots
\label{eq-23}
\end{equation}
The following elementary but 
useful consequence of Lemma~\ref{lemma-2} 
and of the archimedian condition characterizes
the rational series over max-plus like semirings.
This theorem, which is a series analogue
of the max-plus cyclicity theorem
for powers of max-plus matrices
of Cohen, Dubois, Quadrat and Viot~\cite{cohen83}
(see also~\cite{cohen85a,BICOQ,maxplus97,agw04,atwork}), 
was anticipated by Cohen, Moller, Quadrat
and Viot in~\cite{cohen89a}, where
a result similar to Theorem~\ref{th-maxplusrat}
is proved in the special case of series 
with nondecreasing coefficient sequence.
Moller~\cite{moller88},
and Bonnier-Rigny and Krob~\cite{krob93a},
proved results which are
essentially equivalent to Theorem~\ref{th-maxplusrat},
which is taken from~\cite{gaubert92a,gaubert94a,gaubert94d}
(slightly more general assumptions are made on the
semiring, in the last two references).
Theorem~\ref{th-maxplusrat} is in fact a max-plus analogue
of a deeper result, Soittola's theorem~\cite{soittola}, which characterizes
nonnegative rational series
as merges of series with a dominant root
(see also Perrin~\cite{perrin}).
\typeout{FARINA TO BE ADDED}
\begin{theorem}\label{th-maxplusrat}
Let $K$ denote an idempotent linearly ordered
archimedian commutative semiring. A series $S\in K[[X]]$
is rational if, and only if, it is a merge
of ultimately geometric series.
\end{theorem}
\begin{proof}
We have to show that a rational series $S\in K[[X]]$
satisfies 
\begin{equation}
\<S,X^{(k+\kappa)c+j}> =
u q^k \enspace, \quad
\forall k\geq 0,\; 0\leq j\leq c-1\enspace, 
\label{toprove}
\end{equation}
for some $u,q\in K$, and for some integers $\kappa\geq 0, c\geq 1$.
But $S$ has a representation of the form~\eqref{nice},
i.e. $\<S,X^{(k+\kappa_1)c+j}> =\bigoplus_{i\in I_j} u_i q_i^k$,
where $u_i,q_i\in K$, $I_j$ is a finite set,
and $\kappa_1\geq 0,c\geq 1$ are integers.
Let $q=\bigoplus_{i\in I_j} q_i$.
Since $K$ is linearly ordered and $\oplus$ coincides with the least
upper bound, we can find an index $\ell$ such
that $q_\ell=q$, and $u_\ell\geq u_{m}$ for all $m$ such that $q_m=q$.
Then, 
$\<S,X^{(k+\kappa_1)c+j}> =\bigoplus_{i\in I_j,\; q_i<q} u_i q_i^k
\oplus u_\ell q_\ell^k$.
Using the archimedian property, we get
$\<S,X^{(k+\kappa_1)c+j}> = u_\ell q_\ell^k$,
for $k$ large enough, say for $k\geq k_2$.
Setting $\kappa=\kappa_1+k_2$, $q=q_\ell$, and $u=u_\ell q_\ell^{k_2 c}$,
we get~\eqref{toprove}.
\end{proof}
Equivalently, $S$ can be written as 
\begin{equation}
S = P\oplus X^{\kappa c}
\Big(\bigoplus_{0\leq i\leq c-1} u_i X^{i} (q_i X^c)^* 
\Big) \enspace, 
\label{ratu3}
\end{equation}
where $P\in K[X]$ has degree less than $\kappa c$,
and $u_i,q_i\in K$.

Finally, we shall need the inverse operation
of merging, that we call undersampling. 
For each integer $0\leq j\leq c-1$
and for all series $T\in K[[X]]$,
we define the undersampled series:
\[
T^{(j,c)}=
\bigoplus_{k \in \N} \< T, X^{kc+j} > X^k
\enspace. 
\]
For instance, when $T$ is as in~\eqref{eq-23},
$T^{(0,2)}$ and $T^{(1,2)}$
respectively 
coincide with the series $S^{(0)}$ and $S^{(1)}$ of~\eqref{eq-22}.
Trivially, testing the equality of two series
amounts to testing the equality
of undersampled series:
\begin{lemma}\label{lem-3}
Let $c\geq 1$. Two series $\sfT,\sfT'\in \sfK[[X]]$
coincide if, and only if, $\sfT^{(j,c)}=
{\sfT'}^{(j,c)}$ for all $0\leq j\leq c-1$.\qed
\end{lemma}
A last, trivial, remark will allow us to split 
the test that $\sfS=S$ into transient and ultimate parts.
Recall that $X^{-m} \sfS$ denotes the series $\sfT$
such that $\<\sfT, X^k>=\<\sfS,X^{m+k}>$. 
\begin{lemma}\label{lem-4}
Let $m\geq 0$. 
Two series $\sfT,\sfT'\in \sfK[[X]]$
coincide if, and only if, $\<\sfT,X^k>
= \<\sfT',X^k>$ for $k\leq m-1$,
and $X^{-m}\sfT=X^{-m}\sfT'$.\qed
\end{lemma}

\section{Proof of Theorem~\ref{th-2}}
\label{proof-section}
In the sequel, $K$
denotes a semiring that satisfies the assumptions
of Theorem~\ref{th-2} and  $\Sigma=\{\sfd_1,\ldots,\sfd_n\}$ is a finite
set of commuting indeterminates. 
We first prove a simple lemma.
\begin{lemma}\label{lem-7}
For all $\sfp \in K[\Sigma]$ and
$p\in K$, the sets 
$\{\sfp\leq p\}$, 
$\{\sfp \geq p\}$, and
$\{\sfp= p\}$, 
are semi-polyhedral. 
\end{lemma}
\begin{proof}
Since in an idempotent semiring
$u\leq v\iff u\oplus v=v$, 
it suffices to prove more generally
that when $\sfp,\sfq\in K[\Sigma]$,
$\{\sfp=\sfq\}$ is semi-polyhedral.
When $\sfp$ or $\sfq=\zero$,
$\{\sfp=\sfq\}$ is trivially semi-polyhedral.
Otherwise, we can write $\sfp$ and $\sfq$ as finite sums of monomials,
$\sfp=\bigoplus_{i\in I} \sfp_i$, and $\sfq=\bigoplus_{j\in J} \sfq_j$,
with $I,J\neq \emptyset$.
For all $(i,j)\in I\times J$, consider
the polyhedron $U_{ij}=\big(\cap_{k\in I} \{\sfp_i \geq \sfp_k\}\big)
\cap \big(\cap_{l\in J} \{\sfq_j\geq \sfq_l\}\big)
\cap \{\sfp_i=\sfq_j\} $.
Since $K$ is linearly ordered, 
and since the sum $\oplus$ is the least upper bound for $\leq$, 
$\{\sfp=\sfq\}=\cup_{i\in I,j\in J} U_{ij}$
is a semi-polyhedral set.
\end{proof}
The fact that $\{\sfp =p\}$ is semi-polyhedral was 
already noticed by De Schutter~\cite{deschutter96} (who derived
this result by modelling $\sfp=p$ as an extended linear complementarity
problem).

We now prove Theorem~\ref{th-2}
(the proof will be illustrated by the examples in \S\ref{sec-ex}).
The discussion following the proof of
Lemma~\ref{lemma-2} shows that
the rational series $\sfS\in K[\Sigma][[X]]$
can be represented as~\eqref{ratu2}. Let $F(\sfS)$
denote the set of couples of integers $(\kappa,c)$
for which $\sfS$ has such a representation.
The rational identities~\eqref{eq-cyclic},\eqref{ur} imply that
$(\kappa,c)\in F(\sfS)\implies (\kappa,ck)\in F(\sfS)$ for all
$k\geq 1$. Similarly, 
the rational identity~\eqref{ur} shows that
$(\kappa,c)\in F(\sfS)\implies (k,c)\in F(\sfS)$,
for all $k\geq \kappa$. The same argument can be applied
to the set $F'(S)$ of couples of integers
$(\kappa,c)$ for which the rational series
$S\in K[[X]]$ has a representation of the form~\eqref{ratu3}.
Hence, $F(\sfS)\cap F'(S)\neq \emptyset$, 
which allows us to assume that $\sfS$ and $S$
are given by~\eqref{ratu2} and~\eqref{ratu3},
where $\kappa$ and $c$ are the same in both formul\ae.

By Lemma~\ref{lem-3},
$\{\sfS=S\}= \cap_{0\leq j\leq c-1} 
\{\sfS^{(j,c)}=S^{(j,c)}\}$.
Since the intersection of semi-polyhedral
sets is semi-polyhedral,
and since the series $\sfS^{(j,c)}$ and
$S^{(j,c)}$ have expressions of the form~\eqref{ratu2}
and~\eqref{ratu3}, respectively, but with
$c=1$, it suffices to show Theorem~\ref{th-2} when $c=1$.
Moreover, thanks to Lemma~\ref{lem-4},
$\{\sfS=S\}=\cap_{0\leq k\leq \kappa-1} \{\<\sfS,X^k>=\<S,X^k>\}
\cap \{X^{-\kappa}\sfS = X^{-\kappa} S\}$. By Lemma~\ref{lem-7},
the sets $\{\<\sfS,X^k>=\<S,X^k>\}$ are semi-polyhedral,
hence, using again the closure of semi-polyhedral
sets by intersection, it suffices to show that 
$\{X^{-\kappa}\sfS = X^{-\kappa} S\}$
is semi-polyhedral. 
The series $X^{-\kappa}\sfS$ and $X^{-\kappa} S$
again have expressions of the form~\eqref{ratu2}
and~\eqref{ratu3}, respectively, but with $\kappa=0$,
i.e. with $\sfp=\zero$ and $p=\zero$. Summarizing,
it remains to prove Theorem~\ref{th-2}
when
\begin{eqnarray}
\sfS &=& \bigoplus_{1\leq i\leq r} \sfu_i (\sfq_i X)^* \quad\mrm{and}\label{eq-convex}\\
S & = & u (qX)^* \enspace .\label{eq-24}
\end{eqnarray}
It is easy to eliminate the case where
$u=\zero$. Then, by Lemma~\ref{lem-7},
$\{\sfS=S\}=\{\bigoplus_{1\leq i\leq r}\sfu_i=\zero\}$
is semi-polyhedral. When $q=\zero$,
$\{\sfS=S\}=\{\sfS=uX^0\}=\{\<\sfS,X^0>=u\}\cap \{X^{-1}\sfS=\zero\}$
is semi-polyhedral. Thus, in the sequel, we shall
assume that $u,q\neq \zero$.

The reduction to~\eqref{eq-convex}
leads us to studying special series of this
form. We call {\em line} a series of the form $T=u(qX)^*$,
where $u,q\in K\setminus\{0\}$,
and we say that a series $T\in K[[X]]$ 
is {\em convex} if it is a
finite sum of lines. When $K=\qmax$,
$T$ is a line if, and only if, $k\mapsto \<T,X^k>$
is an ordinary (discrete, half-)line,
and $T$ is convex if, and only if,
$k\mapsto \<T,X^k>$ is a finite
supremum of lines.
Convex series already arose
in~\cite{gbc} (where it was shown that
the minimal dimension of realization
of a convex series can be computed in polynomial time,
but here, we must find {\em all} convex
realizations of~\eqref{eq-24}).
Lines can be easily compared:
\begin{lemma}
Let $u,q,v,w\in K$. Then, $v(wX)^*\leq 
u(qX)^*\implies v=\zero$
or ($v\leq u$ and $w\leq q$).
\label{lem-8}
\end{lemma}
\begin{proof}
The inequality $v(wX)^*\leq 
u(qX)^*$ means that $vw^k\leq uq^k$,
for all $k\geq 0$.
If $v\neq \zero$, the archimedian
property implies that $w\leq q$. Moreover,
taking $k=0$, we get $v\leq u$.
\end{proof}
We shall need the following refinement
of the archimedian condition.
\begin{lemma}\label{arch}
For all $\alpha,\beta,\gamma,\delta\in K$, 
\begin{equation}
(\alpha < \beta\;\mrm{and}\; \delta\neq \zero)
\implies \gamma\alpha^k < \delta \beta^k \;\mrm{for $k$ large
enough.}\label{eq-arch2}
\end{equation}
\end{lemma}
\begin{proof}
Since $K$ is linearly
ordered, the archimedian condition
means precisely that 
\begin{equation}
(\alpha < \beta\;\mrm{and}\; \delta\neq \zero)
\implies \gamma\alpha^k < \delta \beta^k \;\mrm{for some $k$.}
\label{eq-arch3}
\end{equation}
Multiplying the inequality $\gamma\alpha^k\leq \delta\beta^k$
by $\beta$, we get $\gamma\alpha^{k+1}\leq
\gamma\alpha^k\beta\leq \delta\beta^{k+1}$. 
Moreover, since $K$
is cancellative and $\beta\neq \zero$
(because $\beta>\alpha\geq \zero$),
$\gamma\alpha^k\beta= \delta\beta^{k+1}$
would imply $\gamma\alpha^k= \delta\beta^{k}$,
which contradicts~\eqref{eq-arch3}. Hence,
$\gamma\alpha^{k+1}< \delta\beta^{k+1}$,
and after an immediate induction on $k$,
we get~\eqref{eq-arch2}.
\end{proof}

The final, critical, step of the proof 
of Theorem~\ref{th-2} is
an observation, which, when specialized
to $K=\qmax$, is
a geometrically obvious fact about ordinary piecewise affine
convex maps. 
\begin{lemma}\label{convex}
Let $u,q\in K$, $S=u(qX)^*$,
$u_1,\ldots,u_r,q_1,\ldots,q_r\in K$,
$T_i=u_i(q_iX)^*$, and $T=\bigoplus_{1\leq i\leq r} T_i$.
Then, $T=S$ if, and only if, 
$T_i \leq S$
for all $1\leq i \leq r$, and 
$T_j=S$
for some $1\leq j\leq r$.
\end{lemma}
\begin{proof}
Since $\oplus$ is the least upper bound
in $K[[X]]$, if $T= S$, we have for all $1\leq i\leq r$,
$T_i\leq S$, which, by Lemma~\ref{lem-8}, means either $u_i=\zero$
or ($(u_i\leq u)$ and $(q_i\leq q)$).
Let $I=\set{1\leq i\leq r}{u_i\neq\zero}$.
We shall assume that $S\neq\zero$, i.e, that $u=\zero$
(otherwise the result is obvious).
Since $T=S$ and $S\neq\zero$, $I\neq\emptyset$.
Now, let $\overline{q}=\bigoplus_{i\in I} q_i\leq q$,
$J=\set{i\in I}{q_i=\ov q}$, and $\ov u=\bigoplus_{j\in J} u_j$.
Using~\eqref{eq-arch2},
we get 
$\<T,X^k>= \ov u \,\ov q^k$,
for $k$ large enough. Identifying this expression with
$\<S,X^k>= uq^k$, and using the archimedian condition,
we get $\ov q=q$. Cancelling $q^k$ 
in $\ov u q^k=u q^k$, we get $\ov u=u$,
and since $K$ is linearly ordered,
$\ov u=\bigoplus_{i\in J} u_i= u_j$ for some $j\in J$. Thus, $S=T_j$,
which shows that the condition of the lemma is necessary.
The condition is trivially sufficient.
\end{proof}
We now complete the proof of Theorem~\ref{th-2}.
Let $\sfS_i=\sfu_i (\sfq_i X)^*$. By Lemma~\ref{lem-8},
both $\{\sfS_i\leq S\}=\{\sfu_i=\zero\}\cup
(\{\sfu_i\leq u\}\cap \{\sfq_i\leq q\})$ 
and $\{\sfS_i=S\}=\{\sfu_i=u\}\cap\{\sfq_i=q\}$
are semi-polyhedral. Hence, by Lemma~\ref{convex},
$\{\sfS= S\}= \cup_{1\leq j\leq r} 
\big(\{\sfS_j =S\}
\cap (\cap_{i\in I, i\neq j} \{\sfS_i\leq S\}
)
\big)$
is semi-polyhedral, which concludes
the proof of Theorem~\ref{th-2}.
\section{Examples}
\label{sec-ex}
\subsection{First example}
Let us illustrate the algorithm of the proof
of Theorem~\ref{th-2} by computing
$\{\sfS=S\}$ when $K=\qmax$,
$\sfS=\sfu_1(\sfv_1 X)^* 
\oplus \sfu_2(\sfv_2 X^2)^*$,
$\Sigma=\{\sfu_1,\sfu_2,\sfv_1,\sfv_2\}$
and $S=0\oplus X(1X)^*
= 0\oplus 0X\oplus 1X^2\oplus 2X^3\oplus \cdots$
The first step of the proof consists
in putting $\sfS$ and $S$ in the forms~\eqref{ratu2}
and~\eqref{ratu3}, respectively. Here,
\begin{eqnarray*}
\sfS &=& \sfu_1 (\unit \oplus \sfv_1 X)(\sfv_1^2 X^2)^*
\oplus \sfu_2 (\sfv_2 X^2)^*\\
S &=& 0 \oplus (X \oplus 1 X^2)(2X^2)^* \enspace. 
\end{eqnarray*}
Then, $\{\sfS=S\}=\{\sfS^{(0,2)}= S^{(0,2)}\}
\cap \{\sfS^{(1,2)}= S^{(1,2)}\}$,
where the undersampled series are given
by
\begin{eqnarray*}
\sfS^{(0,2)} &=& \sfu_1 (\sfv_1^2 X)^*
\oplus \sfu_2 (\sfv_2 X)^*\\
\sfS^{(1,2)} &=& \sfu_1\sfv_1(\sfv_1^2 X)^*\\
S^{(0,2)} &=& 0 \oplus 1X \oplus 3X^2= 0\oplus 1X(2X)^*\\
S^{(1,2)} &=& (2X)^* \enspace .
\end{eqnarray*}
By Lemma~\ref{lem-8}, $\{\sfS^{(1,2)}= S^{(1,2)}\}
= \{\sfu_1\sfv_1 =0\}\cap \{\sfv_1^2=2\}$.
In $\qmax$, 
the unique solution of the equation $v_1^2=2$ is $v_1=1$,
and the unique solution of $u_1\otimes 1=0$ is
$u_1=-1$. Hence, 
\begin{equation}
\{\sfS^{(1,2)}= S^{(1,2)}\}=
\{\sfu_1=-1\}\cap \{\sfv_1=1\}
\enspace .
\label{eq-ex1}
\end{equation}
The series $S^{(0,2)}$ has an expression
of the form~\eqref{ratu3} with $\kappa=1$.
Let us give an expression~\eqref{ratu2} with the same
$\kappa$ for $S^{(0,2)}$:
\begin{eqnarray*}
\sfS^{(0,2)} &=& \sfu_1 \oplus \sfu_1\sfv_1^2X (\sfv_1^2 X)^*
\oplus \sfu_2 \oplus \sfu_2\sfv_2X(\sfv_2 X)^*
\end{eqnarray*}
Thus, $\{\sfS^{(0,2)}=S^{(0,2)}\}
=  \{\<\sfS^{(0,2)},X^0>=\<S^{(0,2)},X^0>\}
\cap
\{X^{-1}\sfS^{(0,2)}=X^{-1} S^{(0,2)}\}
= \{\sfu_1 \oplus \sfu_2 = 0\}
\cap \{ \sfu_1\sfv_1^2 (\sfv_1^2 X)^*
\oplus \sfu_2\sfv_2(\sfv_2 X)^* =
1(2X)^* \} $. Combining this with~\eqref{eq-ex1}
and using Lemma~\ref{lem-8}, we see that
$\{\sfS=S\}$ is the polyhedron
defined by
\[
\sfu_1=-1, \sfv_1=1\;\;\
\sfu_2 =0 ,\; \sfv_2 \leq 1 \enspace.
\]

\subsection{Second example}
Let $\alpha,\beta\in \Q$, and let us look for the realizations
of dimension $2$ of the series
\begin{equation}
S=X^0\oplus \alpha X^2(\beta X)^* 
\enspace. \label{eq-test}
\end{equation}
The proof of Theorem~\ref{th-alg}
requires to find a star height one
representation for the universal rational series
$\sfS_2=\sfc(\sfA X)^*\sfb$.
Such a representation can be obtained for instance
by
using the McNaughton-Yamada algorithm (see~\cite{hopcroft}, Proof of Th.~2.4),
together with the rational identities~\eqref{eq-cyclic}, or directly
from the classical graph interpretation of $\sfc(\sfA X)^*\sfb$.
Setting $\alpha_{ij}=\sfA_{ij}X$,
we easily get:
\begin{eqnarray}
\nonumber \sfS_2&= & (\sfc_2\alpha_{21}\sfb_1\oplus
\sfc_1\alpha_{12}\sfb_2)(\alpha_{11}\oplus\alpha_{22})^*\oplus 
\sfc_2\alpha_{22}^*\sfb_2 \oplus 
\sfc_1\alpha_{11}^*\sfb_1 \oplus \\
&&\alpha_{21}\alpha_{12}(\alpha_{11}\oplus \alpha_{22}
\oplus \alpha_{12}\alpha_{21})^* 
(\sfc_2\alpha_{21}\sfb_1\oplus\sfc_2\sfb_2
\oplus \sfc_1\alpha_{12}\sfb_2
\oplus \sfc_1\sfb_1)\enspace .
\label{eq-s2}
\end{eqnarray}
After applying the algorithm
of the proof of Theorem~\ref{th-2} to~\eqref{eq-s2}
(we do not reproduce
the computations, which are a bit lengthy, but straightforward),
we get that if $\alpha\leq \beta^2$,
all the realizations of $S$ are similar\footnote{We say,
as usual, that two representations $(c,A,b)$ and $(c',A',b')$
are {\em similar} if $c'=cP, A'=P^{-1}AP, b'=P^{-1}b$,
for some invertible matrix $P$. In the max-plus semiring, 
an invertible matrix is the product of a diagonal matrix by a permutation matrix (see e.g.~\cite{BICOQ} for this standard result).
Unlike in conventional algebra, max-plus minimal realizations
are in general not similar.}
to:
\[ 
c=\begin{pmatrix}
\unit & \zero
\end{pmatrix}\quad
A=\begin{pmatrix}
\zero & \alpha\\
\alpha & \beta
\end{pmatrix}\quad
b=\begin{pmatrix}
\unit \\ \zero
\end{pmatrix} \enspace.
\]
If $\alpha>\beta$, then $S$ has no two
dimensional realization. Surprisingly enough,
realizing even a simple series like~\eqref{eq-test}
is not immediate: we do not know
a simpler way to compute the set of dimension $2$ realizations
of $S$.
\subsection{Counter Example}
Let $\Sigma=\{\sfu_1,\sfu_2,\sfv_1,\sfv_2\}$.
We prove that the subset of $\qmax^4$
\begin{eqnarray}
\cals&=&\{\sfu_1(\sfv_1X)^* \oplus \sfu_2 (\sfv_2 X)^*\geq
(0X)^*\}  \nonumber \\
&=& 
\set{(u_1, v_1, u_2, v_2)}{\forall k\in \N,\; \max(u_1+kv_1,
u_2+kv_2)\geq 0 }  \label{e-max}
\end{eqnarray}
is not semi-polyhedral.
It suffices to show that the projection $A$
of $\cals\cap \{\sfv_1=\sfu_2=-1\}\cap \{\sfu_1,\sfv_2 \geq 0\}$
on the coordinates
$\sfu_1,\sfv_2$ is not semi-polyhedral. 
Let $f(k) = \max(u_1-k, -1+v_2k)$.
Specializing~\eqref{e-max}, we see that $(u_1,v_2)\in A$
if, and only if, $u_1,v_2\geq 0$ and $\min_{k\in \N} f(k) \geq 0$.
The map $f$ is decreasing from $0$ to $x=(u_1+1)/(v_2+1)$, and
increasing from $x$ to $+\infty$, therefore, 
$\min_{k\in \N} f(k) \geq 0\iff f(\floor x)\geq 0
\;\mrm{and}\;f(\ceil x)\geq 0$, which gives\footnote{We recall that $\floor{x}$ stands for the integer part of $x$ and $\ceil{x}$ is equal to $-\floor{-x}$ and is the rounding to the smallest bigger than $x$ integer}
\[
 A=\set{(u_1,u_2)}{u_1,v_2\geq 0,\;\; u_1- \floor{\frac{u_1 +1}{v_2+1}}
\geq 0, -1+v_2 
\ceil{\frac{u_1 +1}{v_2+1}} \geq 0}
\enspace .
\]

The set $A$ is depicted by the grey region on Figure~\ref{fig-cex}. 
Note that the border of this region contains an infinite number of vertices
lying on the hyperbola
$u_1v_2=1$. 
\begin{figure}
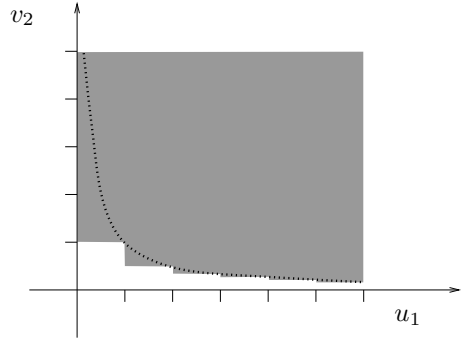

\begin{center}
\input hyperbole3
\end{center}
\caption{The set of realizations $(c,A,b)$ of dimension $2$ such that $cA^kb\geq 0$ for all $k$ is not semi-polyhedral.  A two dimensional
section of this set is represented.}
\label{fig-cex}
\end{figure}
It is geometrically obvious that $A$ is not
semi-polyhedral, but we can check it without
appealing to the figure, as follows.
For any integer $n$ the point $(n,1/n)$ belongs to the set $A$. If $A$ 
was a finite union of polyhedra, then there would
be a polyhedron $P\subseteq A$
with an infinite number of points of $(n,1/n)$ in $P$, and the low
borderline of $P$ would be the line $\{v_2=0\}$. This
is not possible, because for $v_2>0$, the point $(u_1,v_2)$ is not
in $A$, as soon as $v_2<1/(u_1+2)$.

\section{Universal Commutative Rational Expressions and Complexity Analysis}
\label{sec-universal}
In this section, we bound the complexity
of the algorithm of the proofs of Theorem~\ref{th-alg}
and~\ref{th-2}. 
Suppose we are looking for a realization of size $N$ of the series $S$
given as in~\eqref{ratu3}:
\begin{align}\label{e-s0}
S = P\oplus X^{\kappa_0c_0}
\Big(\bigoplus_{0\leq i\leq c_0-1} u_i X^{i} (q_i X^c)^* 
\Big) \enspace, 
\end{align}
where $c_0\geq 1$, $\kappa_0\geq 0$, 
$P\in K[X]$ has degree less than $\kappa_0 c_0$,
and $u_i,q_i\in K$.
A critical step of our algorithm
is to build, like we did in~\eqref{eq-s2}, a star height
one representation of the form~\eqref{ratu2}
for the universal series $\sfS=\sfS_N$:
\begin{align}
\label{e-s1}
\sfS = \sfP\oplus X^{\kappa_1c_1}
\Big(\bigoplus_{1\leq i\leq \rho} \sfu_i X^{\mu_i} (\sfq_i X^{c_1})^* 
\Big)
\end{align}
where $0\leq \mu_i\leq c_1-1$,
$\sfu_i\in \sfK$, and $\sfP\in \sfK[X]$ has degree
less than $\kappa_1 c_1$. 
In section~\S\ref{s-c1}, 
we shall give an explicit
star height one representation for $\sfS_N$
which is of independent interest. This
expression has a double exponential size.
In~\S\ref{s-c2},
we shall bound the size of an expression
of $\{\sfS=S\}$ as a union of intersection
of half-spaces, when $\sfS$ and $S$ are
given by~\eqref{e-s0} and~\eqref{e-s1},
and show that the subproblem of computing $\{\sfS=S\}$ has
a simply exponential complexity.
Finally, in section~\S\ref{s-c3}, we shall
combine the results of~\S\ref{s-c1} and~\S\ref{s-c2}
to show that the method of Theorem~\ref{th-alg}
yields a triply exponential algorithm to compute
the set of realizations of a max-plus
rational series. This triply exponential
bound is a coarse one:
trying examples by hand suggests that the naive
version of the algorithm that we analyse here could
be made much more practicable
by using linear programming and constraint programming techniques.
\subsection{Universal Commutative Rational Expressions}
\label{s-c1}
Let us associate to the triple
$\sfc\in \sfK^{1\times N},\sfA\in \sfK^{N\times N},
\sfb\in \sfK^{N\times 1}$ a digraph $G_N$ composed 
of the nodes $1,\ldots,N$, together
with an {\em input} node $\nin$
and an {\em output} node $\nout$,
arcs $j\to i$ with weights
$\sfA_{ij}X$, for $1\leq i,j\leq N$,
input arcs $\nin \to i$ with weights $\sfb_i$,
and output arcs $i\to \nout$ with weights $\sfc_i$.
The {\em weight} of a path $\pi$, denoted by $w(\pi)$,
is defined as the product of the weight
of its arcs.
We say that two circuits $\gamma $ and $\gamma'$ 
are {\em cyclic conjugates} if one is obtained
from the other by a circular permutation. 
When $\sfK$ is commutative, $w(\gamma)=w(\gamma')$.
We denote by $\calc_N$ the set of conjugacy classes
of elementary circuits of $G_N$.
Let $C\subset \calc_N$, and let $\pi$ denote a path of $G_N$. We say 
that $C$ is {\em accessible} from $\pi$ 
if the union of the circuits of $C$
and of the path $\pi$ is a connected subgraph
(we use here the {\em undirected} notion
of connectedness, not to be confused with strong
connectedness). An accessible set $C$ for a path $\pi$ looks typically
as follows:


\begin{center}
\input path
\end{center}

\noindent
We denote by $\cala(\pi)$ the set of $C\subset \calc_N$ accessible
from $\pi$.  We set $S^+=SS^*$, for all
series $s$ such that $S^*$ is well defined.
We denote by $\calp_N$ the set of elementary paths
from $\nin$ to $\nout$.
The following result is Lemma~6.2.3 from~\cite[Chap.~VII]{gaubert92a}.
\begin{proposition}\label{lem-combi-commut}
Let $\sfK$ denote a commutative idempotent semiring,
$\sfA\in \sfK^{N\times N}$,
 $\sfb\in \sfK^{N\times 1}$,
$\sfc\in \sfK^{1\times N}$,
and $\sfS_N=\sfc(\sfA X)^*\sfb$.
We have
\begin{equation}
\sfS_N  =\bigoplus_{\pi\in \calp_N} w(\pi)\left(
\bigoplus_{C\in \cala(\pi)} \bigotimes_{\gamma\in C} w(\gamma)^+
\right)\enspace .\label{e-l-c}
\end{equation}
\end{proposition}
(By convention, $\emptyset\in \cala(\pi)$ for
all paths $\pi$, and the products in~\eqref{e-l-c}
corresponding to $C=\emptyset$ are equal to $\unit$.)

Before proving Proposition~\ref{lem-combi-commut},
it is instructive to consider the case when $N=2$.
Then, there are four paths
in the sum~\eqref{e-l-c}, 
$\pi_1=\nin\to 1\to \nout$,
$\pi_2=\nin\to 2\to \nout$,
$\pi_3=\nin\to 1\to 2\to \nout$,
$\pi_4=\nin\to 2\to 1\to \nout$,
with respective weights $\sfc_1\sfb_1$,
$\sfc_2\sfb_2$,
$\sfc_2\alpha_{21}\sfb_1$,
and $\sfc_1\alpha_{12}\sfb_2$.
We have for instance 
$\cala(\pi_1)=\{ \{1\to 1\}, \{1\to 2\to 1\},
\{1\to 2\to 1,1\to 1\},
\{1\to 2\to 1,2\to 2\},
\{1\to 2\to 1,1\to 1,2\to 2\}\}$.
Thus, the contribution of $\pi_1$ in~\eqref{e-l-c}
is
\[
\sfc_1\sfb_1(\unit \oplus
\alpha_{11}^+ \oplus (\alpha_{12}\alpha_{21})^+
\oplus \alpha_{11}^+ (\alpha_{12}\alpha_{21})^+
\oplus (\alpha_{12}\alpha_{21})^+ \alpha_{22}^+ 
\oplus \alpha_{11}^+(\alpha_{12}\alpha_{21})^+ \alpha_{22}^+ )
\]
and, considering the similar contributions
of $\pi_2,\pi_3,\pi_4$, it is easy to see that~\eqref{e-l-c}
coincides with~\eqref{eq-s2}.

\begin{proof}[Proof of Proposition~\ref{lem-combi-commut}]
Let $B$ denote the right hand side of
(\ref{e-l-c}). We shall prove
by induction on $k$ the following
property:  ($H_k$) for all (possibly non elementary)
paths $\pi$ from $\nin$ to $\nout$, 
for all sets of $k$ circuits
$C=\{\gamma_1,\ldots , \gamma_k\}\in \cala(\pi)$, 
and for all $n_1,\ldots,n_k\geq 1$,
$w(\pi)w(\gamma_1)^{n_1}\ldots w(\gamma_k)^{n_k}$ is the weight
of a path $\pi'$ from $\nin$ to $\nout$.
If $k=1$, since the graph induced by $\pi\cup \gamma_1$ is connected,
$\gamma_1$ must have a common node with $\pi$, say node $r$.
Possibly after replacing $\gamma_1$ by a cyclic conjugate,
we may assume that $r$ is the initial (and final) node of $\gamma$.
We can write $\pi=\pi_{\nout, r} \pi_{r,\nin}$
(here, and in the sequel, composition of paths is denoted by concatenation),
where $\pi_{r,\nin}$ is a path from $\nin$
to $r$, and $\pi_{\nout, r}$ is a path
from $r$ to $\nout$. Thus
$w(\pi)w(\gamma_1)^{n_1}=w(\pi_{\nout, r}\gamma_1^{n_1}\pi_{r,\nin})$
is the weight of the path $\pi'=
\pi_{\nout, r}\gamma_1^{n_1}\pi_{r,\nin}$
from $\nin$ to $\nout$, which proves $(H_1)$.
We now assume that $k\geq 2$. By definition
of $\cala(\pi)$, at least one of the circuits 
$\gamma_1,\ldots,\gamma_k$, say $\gamma_1$, has a node in common with $\pi$.
Arguing as in the proof of $(H_1)$, 
we see that $w(\pi)w(\gamma)^{n_1}$ is 
the weight of a path $\pi'$ from $\nin$ to $\nout$,
which is such that $\{\gamma_2,\ldots,\gamma_k\}\in \cala(\pi')$.
Applying $(H_{k-1})$ to $\pi'$, we get
$(H_k)$.  

Since $(H_k)$ holds for all $k$,
all the terms of the sum $B$
can be interpreted as weights of paths
from $\nin$ to $\nout$, but we know
that  $\sfS_N$
is the sum of the weights of all paths
from $\nin$ to $\nout$. Hence,
$B\preceq \sfS_N$.
Conversely, if $\pi$ is a path from $\nin$ to $\nout$, we can
write $\pi=\pi_1\gamma_1^{n_1}\pi_2\ldots \gamma_k^{n_k}\pi_{k+1}$,
where $\pi_1\pi_2\ldots \pi_{k+1}$ is an elementary path
from $\nin$ to $\nout$,
and $\gamma_1,\ldots,\gamma_k$ are elementary circuits
which form an accessible set for $\pi$.
This implies that $w(\pi)\preceq B$, and since this holds
for all $\pi$, $\sfS_N\preceq B$. 
\end{proof}
We tabulate the size of the sets
determining the size of the expression~\eqref{e-l-c},
for further use. We denote by $\# X$ the cardinality
of a set $X$.  It is easy to check that 
\begin{align}
\# \calp_N =\sum_{i=0}^N \frac{N!}{i!}
\leq e N!=\calo(N!) \enspace,
\end{align}
and that
\[
\# \calc_N = \sum_{i=1}^N \frac{N!}{(N-i)! i!} \enspace.
\]
The $\calc_N$ are called {\em logarithmic numbers},
their exponential generating function,
$\sum_{N\geq 1} (N!)^{-1}\calc_N z^N$,
is equal to $-\log(1-z)\exp(z)$.
Using for instance a singularity analysis~\cite[Th.~2]{FlOd90b},
we get
\begin{align}
\# \calc_N = \calo((N-1)!\log N) \enspace .
\end{align}
We have of course $\#C \leq \#\calc_N$, and $\#\cala(\pi)
\leq 2^{\#\calc_N}$, for all $C\in \calc_N$ and
$\pi\in \calp_N$.
\subsection{Computing $\{\sfS=S\}$}
\label{s-c2}
We now embark in the complexity analysis of the algorithm contained
in the proof of Theorem~\ref{th-alg}. This analysis involves
a tedious but conceptually simple bookeeping: we 
bound the number of polyhedral sets appearing when
expressing that the star-height one rational expression~\eqref{e-l-c}
evaluates to a given rational series.

If $\sfp\in K[\Sigma]$, we denote by $|\sfp|$ the
number of monomials which appear in $\sfp$
(for instance, if $K=\qmax$, $\Sigma=\{a,b\}$,
$\sfp=\unit \oplus 2a^2\oplus ab\oplus 7b$, $|\sfp|=4$).
We consider the series $\sfS$ and $S$ given by~\eqref{ratu2}
and~\eqref{ratu3}, respectively, with $\sfK=K[\Sigma]$,
and we set
\[
m=\max(\max_{0\leq i<\kappa c} |\<\sfp,X^i>|,
\max_{1\leq i\leq \rho}\max(|\sfu_i|,|\sfq_i|))
\]
and $\rho_i=\#\set{1\leq j\leq
\rho}{\mu_j=i}$.
\begin{proposition}\label{prop-complex}
Let $\sfS$ and $S$ be given by~\eqref{ratu2} and~\eqref{ratu3}, respectively. 
The set $\{\sfS =S\}$ can be expressed as the union
of at most $m^{\kappa c+2c}(\prod_{0\leq i<c}\rho_i)2^{\rho-c}$
intersections of at most $(m+1)\kappa c + 2c+2\rho m$ half-spaces.
\end{proposition}
To show Proposition~\ref{prop-complex},
we need to introduce some adapted notation.
We say that a couple of positive integers
$[n,k]$ is a {\em symbol} of a subset $\cals$ of $K^{\Sigma}$,
and we write $\cals\in [n,k]$,
if $\cals$
can be written as the union of $n$ sets, $\cals=\cup_{1\leq i\leq n} \cals_i$,
where each $\cals_i$ is the intersection of at most $k$ half-spaces.
Of course, $\cals\in [n,k]\implies \cals\in [n',k']$, for
all $n'\geq n,k'\geq k$.
For instance, taking $\sfp=\bigoplus_{i\in I} \sfp_i\in K[\Sigma]$
and $p\in K$ as in Lemma~\ref{lem-7}, and specializing
the proof of Lemma~\ref{lem-7}, we have
\begin{equation}
\{\sfp=p\}= \bigcup_{i\in I} ( \{\sfp_i\leq p\}
\cap \{\sfp_i\geq p\}\cap \bigcap_{j\in I\atop j\neq i} \{\sfp_j\leq p\})
\enspace.
\label{c0}
\end{equation}
Since, by definition, $|\sfp|=\# I$
we get from~\eqref{c0}:
\begin{equation}
\{\sfp=p\}\in [|\sfp|,|\sfp|+1] \enspace .
\label{c1}
\end{equation}

Similarly, since $\{\sfp\leq p\}=\cap_{i\in I}
\{\sfp_i\leq p\}$,
\begin{equation}
\{\sfp\leq p\}\in [1,|\sfp|] \enspace .
\label{c1bis}
\end{equation}

It will be convenient to equip symbols with the binary
laws $\scup$ and $\scap$, defined by:
\[
[n,k]\scup [n',k']= [n+n',\max(k,k')],\qquad
[n,k]\scap [n',k']= [nn',k+k'] \enspace.
\]
This notation is motivated by the following rule,
which holds for all subsets $\cals, \cals' \subset K^{\Sigma}$:
\begin{eqnarray}\nonumber
(\cals \in [n,k]\;\mrm{and}\;\cals '\in [n',k'])
&\implies& 
(\cals \cup \cals ' \in [n,k]\scup [n',k']\\
&& \;\mrm{and}\;\;
\cals \cap \cals '\in [n,k]\scap [n',k'] )\enspace .
\end{eqnarray}

\begin{proof}[Proof of Proposition~\ref{prop-complex}]
As a preliminary step, we compute symbols for
the more elementary sets involved in the proof of Theorem~\ref{th-2}.

First, if $\sfu,\sfq\in K[\Sigma]$
and $u,q\in K$, we get from
Lemma~\ref{lem-8},
$\{\sfu(\sfq X)^*\leq u(qX)^*\}=
\{\sfu=\zero \}\cup (\{\sfu\leq u\} \cap \{\sfq\leq q\})$, hence
\begin{equation}
\{\sfu(\sfq X)^*\leq u(qX)^*\}
\in 
[1,|\sfu|] \scup ([1,|\sfu|] \scap [1,|\sfq|])
= [2,|\sfu|+|\sfq|]
\label{c2}
\enspace .
\end{equation}
Moreover, Lemma~\ref{lem-8} shows that $\{\sfu(\sfq X)^*=u(qX)^*\}
=\{\sfu=u\}\cap\{\sfq=q\}$, if $u\neq \zero$.
When $u=\zero$, $\{\sfu(\sfq X)^*=
u(qX)^*\}=\{\sfu=\zero\}=\{\sfu\leq\zero\}$. 
Hence, using~\eqref{c1},~\eqref{c1bis}, we get
\begin{align}
\{\sfu(\sfq X)^*= u(qX)^*\}
\in 
\begin{cases}{}
{}[1,|\sfu|]  \ \ \ \mrm{if $u=\zero$} &\\
[|\sfu|,|\sfu|+1] \scap 
[|\sfq|,|\sfq|+1] = [|\sfu||\sfq|, |\sfu|+|\sfq|+2]&\mrm{else}\\
\end{cases}
\label{c3}
\end{align}

Let us now take $\sfu_1,\ldots,\sfu_r,
\sfq_1,\ldots,\sfq_r\in K[\Sigma]$, $u,q\in K$,
$\sfT_i=\sfu_i(\sfq_i X)^*$, $S =u(qX)^*$.
Lemma~\ref{convex} shows that
\begin{equation}
\Big\{\bigoplus_{1\leq i\leq r} \sfT_i = S\Big\}
= 
\bigcup_{1\leq i\leq r} \big( \{\sfT_i= S\}
\cap
\bigcap_{1\leq j\leq r\atop j\neq i} \{\sfT_j\leq S\}\big)
\enspace,
\end{equation}
hence, using~\eqref{c2} and~\eqref{c3}
\begin{eqnarray}\nonumber
\Big\{\bigoplus_{1\leq i\leq r} \sfT_i = S \Big\}
&\in& \scup_{1\leq i\leq r}
\big(
[|\sfu_i||\sfq_i|,|\sfu_i|+|\sfq_i|+2]
\scap \scap_{1\leq j\leq r\atop j\neq i}
[2, |\sfu_j|+|\sfq_j|]\big)\\
&=& 
(\sum_{1\leq i\leq r} 
|\sfu_i||\sfq_i| 2^{r-1},
2+ \sum_{1\leq i\leq r} (|\sfu_i|+|\sfq_i|))
\label{c4}
\end{eqnarray}
The proof of Theorem~\ref{th-2},
together with~\eqref{c4},~\eqref{c0},
shows that 
\begin{eqnarray}
\nonumber
\{\sfS=S\} &=& \bigcap_{0\leq i<\kappa c} \{\sfp_i=p_i\}\cap 
\bigcap_{0\leq i< c} \Big\{\bigoplus_{1\leq j\leq \rho\atop \mu_j=i}
\sfu_j(\sfq_j X)^* =u_i (q_i X)^*\Big\}\\
&\in& 
\scap_{0\leq i<\kappa c} [m,m+1]
\scap \scap_{0\leq i< c}
[\rho_i m^2 2^{\rho_i-1}, 2+ 2\rho_i m]\nonumber\\
&=& [m^{\kappa c+2c}(\prod_{0\leq i<c}\rho_i)2^{\rho-c},
(m+1)\kappa c + 2c+2\rho m]
\enspace ,
\label{c7}
\end{eqnarray}
which concludes the proof.
\end{proof}

\subsection{Final Complexity Analysis}
\label{s-c3}

Let $\sfK=K[\Sigma]$ and $E$ be a formal expression of a polynomial
$\sfP \in \sfK[X]$. We denote by $m(E)$ the maximum
number of monomials of $K[\Sigma]$
arising as a coefficient of a power of $X$ in some polynomial
expression of $E$. By abuse of
notation we will write $m(\sfP)$ instead of $m(E)$. 
For instance, with $\Sigma=\{a,b\}$
and $K=\qmax$, if $\sfS=7aX\oplus 3abX
\oplus X^2(\unit \oplus 8a^2b X)(3X^2)^*$, $m(\sfS)=2 (=|\<\sfP,X>|=
|7a\oplus 3ab|)$.
If the expression is~\eqref{ratu2}:

\begin{equation}
\label{e-d-l}
m(\sfS)=\max(\max_{0\leq i<\kappa c} |\<\sfP,X^i>|,
\max_{1\leq i\leq \rho}\max(|\sfu_i|,|\sfq_i|))
\end{equation}

\begin{corollary}
\label{cor-univ-expr}
Let $\sfK=K[\Sigma]$, 
$\sfA\in \sfK^{N\times N}$,
 $\sfb\in \sfK^{N\times 1}$,
$\sfc\in \sfK^{1\times N}$,
and $\sfS_N=\sfc(\sfA X)^*\sfb$.
Then  $\sfS_N$ can be written as in \ref{ratu2}:

\begin{equation}
\label{eq-univ-expr}
\sfS_N = \sfP\oplus X^{\kappa_1 c_1}
\Big(\bigoplus_{1\leq i\leq \rho} \sfu_i X^{\mu_i} (\sfq_i X^{c_1})^* 
\Big)
\end{equation}
where $c_1=N!$, $\kappa_1=\mathcal{O}(N!)$, $\rho=2^{\mathcal{O}(N!)}$,
$0\leq \mu_i\leq c_1 -1$, 
$\sfu_i\in \sfK$, $m(\sfS_N)=2^{\mathcal{O}(N!)}$ and $\sfP\in
\sfK[X]$ has degree smaller than $\kappa_1c_1$. 
\end{corollary}

\begin{proof}
We have:
$$\sfS_N  =\bigoplus_{\pi\in \calp_N, C\in \cala(\pi)} 
w(\pi) \bigotimes_{\gamma\in C}  w(\gamma)w(\gamma)^*
$$

\noindent
For every $\gamma$ (and $\pi$ also) the monomial $\sfP=w(\gamma)$ has
degree at most $N$ and is equal to $\sfq X^{\alpha}$ where $\sfq\in \sfK$ is
a monomial (i.e. $m(\sfq =1$) and
$\alpha \leq N$. Let $\alpha '$ be the integer such that $\alpha
\alpha ' = N!$. Using the identity~\eqref{eq-cyclic} we get:

\begin{eqnarray*}
\sfP(\sfq X^{\alpha})^* & = & \sfP (\unit \oplus \sfq X^{\alpha} \oplus \ldots
\oplus \sfq^{\alpha' -1} X^{\alpha(\alpha' -1)})(\sfq^{\alpha'} X^{N!})\\
& = & \sfP' (\sfq^{\alpha'} X^{N!})^*
\end{eqnarray*}

\noindent
where the polynomial $\sfP'$ has degree
$N+N!-\alpha=\mathcal{O}(N!)$ and $m(\sfP')=m(\sfq^{\alpha'})=1$. Using
the identity~\eqref{eq-1} we have immediately: 

$$\sfS_N = \bigoplus_{1\leq j\leq r} \sfP_j (\sfq_j X^{N!})^* $$

\noindent
where $\sfP_1, \ldots, \sfP_r\in \sfK[X]$, $\sfq_1, \ldots, \sfq_r\in \sfK$,
$m(\sfq_j) \leq \# \calc_N$ and $m(\sfP_j)=(N+N!-\alpha)^{\# \calc_N -1}$.  
To evaluate $r$ and the degrees of these polynomials, results from
\S\ref{s-c1} are useful: the degree of each $\sfP_j$ is
$\mathcal{O}((N!)^2)$, 
$r=2^{\mathcal{O}(N!)}$, $m(\sfq_j)=\mathcal{O}(N!)$ and
$m(\sfP_j)=2^{\mathcal{O}(N!)}$.  

Next step is to obtain an expression like~\eqref{ratu2} using the
identity~\eqref{ur}. Here is an example:

\begin{eqnarray*}
s & = & (0 \oplus X^3)(2 X^2)^* \oplus (2 \oplus X^4)(3 X^2)^*\\
& = & (2 X^2)^* \oplus X X^2 (2 X^2)^* \oplus 2 (3 X^2)^* \oplus
X^2 X^2 (3 X^2)^*\\  
& = & (0 \oplus 2 X^2 \oplus 4 X^4 (2 X^2)^*) \oplus X X^2 (0
\oplus 2 X^2(2 X^2)^*) \\
& & \oplus  
2(0 \oplus 3 X^2 \oplus 6 X^4 (3 X^2)^*) \oplus X^2 X^2 (0 \oplus
3 X^2(3 X^2)^*)\\  
& = & (2 \oplus 5 X^2 \oplus X^3\oplus X^4) \ \oplus  \ X^4 (4 (2
X^2)^* \oplus 2X(2 X^2)^* \oplus 8(3 X^2)^* \\
& & \oplus  3X^2 (3 X^2)^*)
\end{eqnarray*}

Let $c_1=N!$ and $\kappa_1$ be the smaller integer such that $c_1
\kappa_1+c_1-1$ is larger than the degrees of $P_j$. Then
$\kappa_1=\mathcal{O}(N!)$ and there are some polynomials $P_{j,0}, 
\ldots, P _{j,\kappa_1}$ of degrees at most $c_1-1$ such that:
$$
P_j= P_{j,0}+X^{c_1} P_{j,1}+X^{2c_1} P_{j,2}+ \ldots +X^{\kappa_1 c_1} P_{j,\kappa_1}
$$
Using~\eqref{ur} we have:

$$\begin{array}{l}
P_j(Q_j X^{c_1})^* = \\
P_{j,0} (\unit \oplus Q_j X^{c_1}\oplus \ldots \oplus
Q_j^{\kappa_1-1} X^{(\kappa_1-1){c_1}} 
\oplus (Q_j^k)_1 X^{\kappa_1 c_1} (Q_j X^{c_1})^* )\\
 \oplus  P_{j,1} X^{c_1}(\unit \oplus Q_j X^{c_1}\oplus \ldots \oplus
Q_j^{\kappa_1-2} X^{(\kappa_1-2)c_1} 
\oplus Q_j^{\kappa_1-1} X^{(\kappa_1-1)c_1} (Q_j X^{c_1})^* ) \\
  \oplus  \ldots  \\
 \oplus P_{j,\kappa_1} X^{\kappa_1 c_1} (Q_j X^{c_1})^* 
\end{array}$$
and thus

$$
P_j(Q_j X^{c_1})^* =R_j \oplus X^{\kappa_1 c_1} \Big(\bigoplus_{0 \leq i\leq
c_1(\kappa_1+1)} 
\sfu_{i,j}  X^{\mu_{i,j}} (Q_j X^{c_1})^*  \Big)
$$
where the degree of the polynomial $R_j$ is at most $\kappa_1 c_1$, the
$\sfu_{i,j}$'s are elements of $\sfK$, $\mu_{i,j}<c_1$, $m(R_j)$
and $m( \sfu_{i,j})$ are $2^{\mathcal{O}(N!)}$. 
At the end we have the equation~\eqref{eq-univ-expr}
where $0\leq \mu_i\leq c_1-1$,
$\sfu_i\in \sfK$, $\sfp\in \sfK[X]$ has degree
less than $\kappa_1 c_1$ and $\rho=c_1(\kappa_1+1)r=2^{\mathcal{O}(N!)}$.

\end{proof}

\begin{corollary}
\label{cor-bonnes-expressions}
Let $S$ be given by~\eqref{ratu3},
$\sfA\in \sfK^{N\times N}$,
 $\sfb\in \sfK^{N\times 1}$,
$\sfc\in \sfK^{1\times N}$,
and $\sfS_N=\sfc(\sfA X)^*\sfb$.
Then  we have

\begin{eqnarray*}
\sfS_N &=& \sfP_2 \oplus X^{k_2 c_2}
\Big(\bigoplus_{1\leq i\leq \rho_2} \sfv_i X^{\mu_i} (\sfr_i X^{c_2})^* 
\Big)\\
S &=& P_2 \oplus X^{k_2 c_2}
\Big(\bigoplus_{0\leq i\leq c_2-1} v_i X^{i} (r_i X^{c_2})^* 
\Big)
\label{last-eq-univ-expr}
\end{eqnarray*}
where $c_2=\lcm(N!, c)$, $k_2=\mathcal{O}(N!)$ and $k_2\leq \kappa_0$,
$\rho_2=c_0 2^{\mathcal{O}(N!)}$,
$0\leq \mu_i\leq c_2 -1$, $v_i \in K$, $\sfv_i\in \sfK$, $\sfP_2 \in
\sfK[X]$ and $P_2 \in K[X]$ have degree smaller than $k_2 c_2$,
$m(\sfS_N)=2^{\mathcal{O}(N!)}$.  
\end{corollary}

\begin{proof}
Considering the equation~\eqref{eq-univ-expr} of
Corollary~\ref{cor-univ-expr} let $c_2=\lcm(c, 
c_1)=\alpha_0 c =\alpha_1 c_1$ and $k_2 =\max \{ \lceil \kappa_0 /
\alpha_0 \rceil, \lceil \kappa_1 /\alpha_1 \rceil \}$. Then we have two
integers $h_0$ and $h_1$ such that $k_2 c_2 = \kappa_0 c + h_0 c= \kappa_1 c_1 +
h_1 c_1$. Using the following equation 

$$
u_i X^i (q_i X^c)^* = u_i X^i (\unit \oplus q_i X^c \oplus \ldots \oplus
q_i^{h_0-1} X^{c(h_0-1)} \oplus q_i^{h_0} X^{c h_0} (q_i X^c)^*)
$$
we have

$$
S = P_2 \oplus X^{k_2 c_2}
\Big(\bigoplus_{0\leq i\leq c-1} u'_i X^{i} (q_i X^{c})^* \Big)
$$
with $u'_i= u_iq_i^{h_0}$ and
$
P_2= P \oplus X^{\kappa_0 c} (\bigoplus_{0\leq i\leq c-1} u_i X^{i})(\unit
\oplus q_i X^c \oplus \ldots \oplus q_i^{h_0-1} X^{c(h_0-1)})
$ is a polynomial of degree at most $k_2 c_2$. Similarly we have:

$$
\sfS_N = \sfP_2 \oplus X^{k_2 c_2}
\Big(\bigoplus_{0\leq \mu_i \leq \rho} \sfu'_i X^{\mu_i} (\sfq_i
X^{c_1})^* \Big) 
$$
where $\sfP_2$ is a polynomial of degree at most $k_2 c_2$ and
$\sfu'_i$ are elements of $\sfK$. The last step is to use the
equation~\eqref{eq-cyclic}

$$
(q_i X^{c})^*=(\unit \oplus q_i X^{c} \oplus \ldots \oplus
q_i^{\alpha_0 -1} X^{c(\alpha_0 -1)})(q_i^{\alpha_0} X^{c_2})^*
$$
which gives us 

\begin{eqnarray*}
S &=& P_2 \oplus X^{k_2 c_2}
\Big(\bigoplus_{0\leq i\leq c_2-1} v_i X^{i} (r_i X^{c_2})^* 
\Big)
\end{eqnarray*}
and similarly

\begin{eqnarray*}
\sfS_N &=& \sfP_2 \oplus X^{k_2 c_2}
\Big(\bigoplus_{1\leq i\leq \rho_2} \sfv_i X^{\mu_i} (\sfr_i X^{c_2})^* 
\Big)
\end{eqnarray*} 
where $r_i=q_i^{\alpha_0}$, $\sfr_i=\sfq_i^{\alpha_1}$,
$v_i=\sum_{j=0}^{\inf\{i, c-1\}} u'_j 
q_i^{i-j} \in K$, $\sfv_i \in \sfK$ and $\rho_2=\rho(1+c_1(\alpha_1
-1))=\rho c_0 \mathcal{O}(N!)$. The bound for $m(\sfS_N)$ follows
easily from Corollary~\ref{cor-univ-expr}.  
\end{proof}

We say that a quantity $Q$ depending of parameters
is {\em simply} (resp. {\em doubly},  {\em triply})
{\em exponential} if $Q$ can be bounded from above by a term
of the form $2^P$ (resp. $2^{2^{P}}, 2^{2^{2^P}}$),
where $P$ is a polynomial function of the parameters.

\begin{corollary}
Let
$S$ be given by~\eqref{ratu3}.
The set of realizations
of dimension $N$ of $S$ can be written as the union
of $n$ intersections of at most $k$ half-spaces, where
$n$ is triply exponential in $N$ and simply exponential in $\kappa,c$, and
$k$ is doubly exponential in $N$ and linear in $\kappa,c$.
In particular, when $K=\qmax$, the existence
of a realization of dimension $N$ of $S$ can be decided
in triply exponential time in $N$ and simply exponential time in
$\kappa,c$. 
\end{corollary}

\begin{proof}
The first statement of the
corollary follows by applying Corollary~\ref{cor-bonnes-expressions},
Proposition~\ref{prop-complex} and using Stirling's formula.
The second statement follows
from the first one together with the fact
that linear programming has a polynomial time complexity
(see e.g.~\cite[Ch.~14 and 15]{schrijver}).
\end{proof}

\section{Conclusion}
We showed that the existence of a realization
of a given dimension of a max-plus linear sequence is decidable,
answering to a question which was raised from the beginning
of the development of the max-plus modelling of discrete event systems,
see~\cite{cohen85b-inria,olsder86b,BICOQ,olsder,bv}. 
This decidability result is obtained as a byproduct of a general structural
property: the set of realizations can be effectively written
as a finite union of polyhedra, or
as the max-plus analogue of a semi-algebraic set.
The complexity analysis leads
to a coarse triple
exponential bound, but it also
shows that some special structured instances of the problem can be 
solved in a more reasonable simply exponential time. A possible 
source of suboptimality of the present bound is that the underlying max-plus
semi-algebraic structure is not exploited: this raises issues of an
independent interest which we will examine further elsewhere.

\bibliography{algebraic}
\bibliographystyle{myalpha}
\end{document}

%% file: hyperbole3.tex
\begin{picture}(0,0)%
\includegraphics{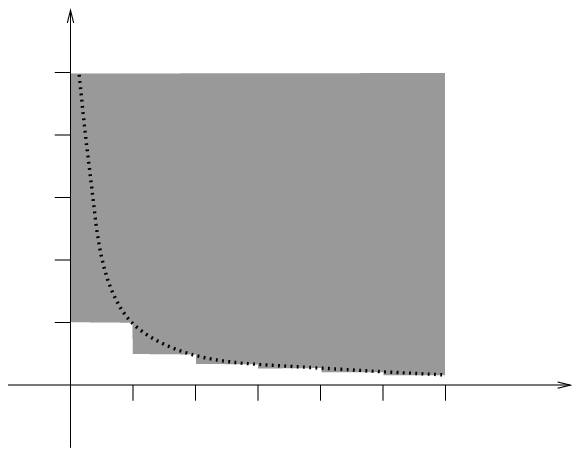}%
\end{picture}%
\setlength{\unitlength}{1973sp}%
\begingroup\makeatletter\ifx\SetFigFont\undefined%
\gdef\SetFigFont#1#2#3#4#5{%
  \reset@font\fontsize{#1}{#2pt}%
  \fontfamily{#3}\fontseries{#4}\fontshape{#5}%
  \selectfont}%
\fi\endgroup%
\begin{picture}(5652,4272)(1576,-7573)
\put(6406,-7351){\makebox(0,0)[lb]{\smash{\SetFigFont{10}{12.0}{\rmdefault}{\mddefault}{\itdefault}{\color[rgb]{0,0,0}$u_1$}%
}}}
\put(1576,-3586){\makebox(0,0)[lb]{\smash{\SetFigFont{10}{12.0}{\rmdefault}{\mddefault}{\itdefault}{\color[rgb]{0,0,0}$v_2$}%
}}}
\end{picture}

%% file: path.tex
\begin{picture}(0,0)%
\includegraphics{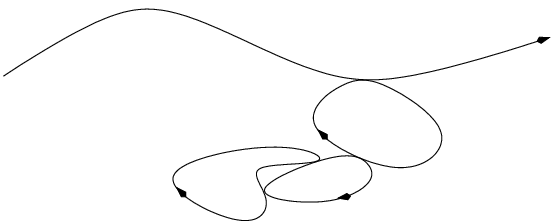}%
\end{picture}%
\setlength{\unitlength}{1973sp}%
\begingroup\makeatletter\ifx\SetFigFont\undefined%
\gdef\SetFigFont#1#2#3#4#5{%
  \reset@font\fontsize{#1}{#2pt}%
  \fontfamily{#3}\fontseries{#4}\fontshape{#5}%
  \selectfont}%
\fi\endgroup%
\begin{picture}(5274,2116)(1789,-6559)
\put(6653,-4659){\makebox(0,0)[lb]{\smash{\SetFigFont{10}{12.0}{\rmdefault}{\mddefault}{\itdefault}{\color[rgb]{0,0,0}$\pi$}%
}}}
\put(6053,-5641){\makebox(0,0)[lb]{\smash{\SetFigFont{10}{12.0}{\rmdefault}{\mddefault}{\itdefault}{\color[rgb]{0,0,0}$c_1$}%
}}}
\put(5356,-6391){\makebox(0,0)[lb]{\smash{\SetFigFont{10}{12.0}{\rmdefault}{\mddefault}{\itdefault}{\color[rgb]{0,0,0}$c_2$}%
}}}
\put(3225,-5828){\makebox(0,0)[lb]{\smash{\SetFigFont{10}{12.0}{\rmdefault}{\mddefault}{\itdefault}{\color[rgb]{0,0,0}$c_3$}%
}}}
\end{picture}